\documentclass[11pt]{article}
\usepackage{rheaj}
\usepackage{setspace}
\usepackage{xparse}

\newcommand{\mypara}[1]{\medskip \noindent {\bf #1}}

\newcommand{\den}{\textnormal{density}}

\begin{document}
\title{A Polylogarithmic Approximation for Directed Steiner
  Forest \\
  in Planar Digraphs}
\author{Chandra Chekuri \thanks{Dept. of Computer Science, Univ. of
  Illinois, Urbana-Champaign, Urbana, IL 61801. {\tt
  chekuri@illinois.edu}. Supported in part by NSF grants CCF-1910149,
CCF-1907937, and CCF-2402667.}
\and
Rhea Jain \thanks{Dept. of Computer Science, Univ. of Illinois,
  Urbana-Champaign, Urbana, IL 61801. {\tt
    rheaj3@illinois.edu}. Supported in part by NSF grants CCF-1907937  and CCF-2402667.}
}
\date{}
\maketitle
\thispagestyle{empty}

\begin{abstract}
\label{sec:abstract}
We consider Directed Steiner Forest (DSF), a fundamental problem in
network design. The input to DSF is a directed edge-weighted graph
$G = (V, E)$ and a collection of vertex pairs
$\{(s_i, t_i)\}_{i \in [k]}$. The goal is to find a minimum cost
subgraph $H$ of $G$ such that $H$ contains an $s_i$-$t_i$ path for
each $i \in [k]$. DSF is NP-Hard and is known to be hard to
approximate to a factor of $\Omega(2^{\log^{1 - \eps}(n)})$ for any fixed
$\eps > 0$ \cite{dodis_design_1999}. DSF admits approximation ratios
of $O(k^{1/2 + \eps})$ \cite{chekuri_set_2011} and
$O(n^{2/3 + \eps})$ \cite{BermanBKRY13}.

In this work we show that in planar digraphs, an important and useful
class of graphs in both theory and practice, DSF is much more
tractable.  We obtain an $O(\log^6 k)$-approximation algorithm
via the junction tree technique.
Our main technical contribution is to prove the
existence of a low density junction tree in planar digraphs. To find an
approximate junction tree we rely on recent results on rooted directed
network design problems \cite{FriggstadM23,chekuri_directed_2024}, in
particular, on an LP-based algorithm for the Directed Steiner Tree
problem \cite{chekuri_directed_2024}. Our work and several other recent ones on algorithms for
planar digraphs
\cite{FriggstadM23,kawarabayashi_embeddings_2021,chekuri_directed_2024} are
built upon structural insights on planar graph reachability and
shortest path separators \cite{thorup_compact_2004}.
\end{abstract}

\newpage 
\setcounter{page}{1}

\section{Introduction}
\label{sec:intro}

Network design is a rich field of study in algorithms and discrete
optimization. Problems in this area are motivated by various practical
applications and have also been instrumental in the development of
important tools and techniques.
Two fundamental problems in network design are Steiner
Tree and Steiner Forest. In Steiner Forest, the input is a graph
$G = (V, E)$ with non-negative edge costs $c: E \rightarrow
\mathbb{R}_+$ and a collection of vertex pairs $D = \{(s_i, t_i)\}_{i \in [k]}$; the goal
is to find a subgraph $H \subseteq G$ of minimum cost such that for
each $i \in [k]$, there exists an $s_i$-$t_i$ path in $H$.  The
Steiner Tree problem is a special case of Steiner Forest in which
there exists some root $r \in V$ such that $s_i = r$ for all
$i \in [k]$; in other words, the goal is to connect a single source
to a given set of sink vertices.  When the input graph is
\emph{undirected}, these problems are both NP-Hard and APX-hard to
approximate, and also admit constant-factor approximation algorithms; see
\Cref{sec:rel-work} for a detailed discussion.

This paper considers the setting in which the input graph is
\emph{directed}.  In Directed Steiner Forest (DSF), the input is a
directed graph $G=(V,E)$ with edge-costs $c: E \rightarrow
\mathbb{R}_+$, and the goal is to find a
min-cost subgraph that contains a \emph{directed path} (dipath) from
$s_i$ to $t_i$ for each $i \in [k]$. Directed Steiner Tree (DST) is
the special case when there is a single source $r$ that needs to be
connected to the sinks $t_i, i \in [k]$. In many settings
directed graph problems tend to be more difficult to
handle. Unlike their undirected counterparts, DST and DSF have
strong lower bounds on their approximability.
DST is known to be hard to approximate to a factor of
$\Omega(\log^{2 - \eps}(k))$ unless NP has randomized quasi-polynomial
time algorithms \cite{HalperinK03}, and to a factor
$\Omega(\log^2 k/\log \log k)$ under other complexity assumptions
\cite{GrandoniLL22} (see \Cref{sec:rel-work}).  Furthermore, a natural
cut-based LP relaxation (used in approximation algorithms for several
undirected network design problems) has a polynomial factor integrality gap;
$\Omega(\sqrt k)$ \cite{ZosinK02} or $\Omega(n^{\delta})$ for some
fixed $\delta > 0$ \cite{LiL22}.  The best known approximation ratios
for DST are $O(k^\eps)$ for any fixed $\eps > 0$ in polynomial time
\cite{Zelikovsky97} and $O(\log^2 k/ \log \log k)$ in
\emph{quasi-polynomial} time \cite{GrandoniLL22}.  
Whether DST admits a polynomial time poly-logarithmic approximation
ratio has remained a challenging open problem for over 25
years.

The Directed Steiner Forest problem, on the other hand, is known
\emph{not} to admit a polylogarithmic approximation ratio unless $\PTIME =
\NP$. This is because
DSF is hard to approximate to a factor
$\Omega(2^{\log^{1 - \eps}(n)})$ for any $\eps > 0$ via a simple
reduction from the Label-Cover problem \cite{dodis_design_1999}.
One technical reason for this difficulty is that, in spite of the name, a
minimal feasible solution to DSF may not be a forest (unlike the case
of undirected graphs). Note the contrast here to DST, in which a
minimal feasible solution is an out-tree rooted at $r$.
Due to this lack of structure as well as aforementioned hardness
results, there has been limited progress in the development of approximation
algorithms.
The current
best approximation ratios for DSF are $O(k^{1/2 + \eps})$
\cite{chekuri_set_2011} in the regime when $k$ is small, 
and $O(n^{2/3 + \eps})$ \cite{BermanBKRY13} when $k$ is large;
both results were obtained over a decade ago.

Recently, Friggstad and Mousavi \cite{FriggstadM23} made exciting
progress on DST by obtaining a simple $O(\log k)$-approximation in
\emph{planar} digraphs. Note that this establishes a separation
between the hardness of DST in general digraphs and in planar
digraphs. A follow-up work by Chekuri et al. \cite{chekuri_directed_2024}
shows that the
cut-based LP relaxation for DST has an integrality gap of
$O(\log^2 k)$ in planar digraphs, which is in sharp contrast to the
known lower bounds in general digraphs.  Motivated by these positive
results, as well as the inherent practical interest of planar graphs,
we consider Directed Steiner Forest in planar digraphs (planar-DSF),
and prove that it admits a poly-logarithmic approximation ratio.

\begin{theorem}
\label{thm:main}
  There is an $O(\log^6 k)$-approximation for Directed Steiner Forest in planar
  digraphs, where $k$ is the number of terminal pairs.
\end{theorem}
\begin{remark}[Node Weights]
  Our algorithm and analysis generalize relatively easily to the
  setting in which both edges and nodes have non-negative
  weights. This is a consequence of the technique. The standard
  reduction of node-weighted problems to edge-weighted problems in
  directed graphs does not necessarily preserve planarity.
  Node-weighted Steiner problems have also been considered separately
  in the undirected setting; see \Cref{sec:rel-work}.
\end{remark}

\subsection{Technical Overview and Outline}
\label{sec:tech-overview}

We prove \Cref{thm:main} by employing the so-called \emph{junction tree} technique.
This technique allows one to reduce a multicommodity problem (such as
DSF) to its rooted/single-source counterpart (such as DST).
The power of this technique comes from
the fact that in several network design problems, the single-source
problem is often easier to solve. Junction-based schemes were initially
highlighted in the context of non-uniform buy-at-bulk network design
\cite{hajiaghayi_approximating_2009,chekuri_approximation_2010},
although the basic idea was already implicitly used for DSF in
\cite{charikar_approximation_1999}. The technique has since been used to make progress in a
variety of network design problems
including improvements to DSF (see \Cref{sec:rel-work}).

The high-level idea is as follows:
we say $H \subseteq G$ is a \emph{partial solution} if it is a feasible solution
for some subset of terminal pairs. We look for a low-density partial solution, where
\emph{density}  is defined as ratio of the cost of
the partial solution to the number of terminal pairs it contains. Using a standard
iterative approach for covering problems, one can reduce the original problem to the
min-density partial solution problem while losing an additional $O(\log k)$ factor in
the approximation ratio. In general, finding a min-density partial solution may
still be hard; therefore, we restrict our attention to well-structured solutions.
We aim to find partial solutions that contain some ``junction'' vertex $r \in V$
through which many pairs connect. Formally, in the context of
DSF, we say a \emph{junction tree} on terminal pairs
$D_H \subseteq D$ is a subgraph $H \subseteq G$ with a root $r$ such that for
every terminal pair $(s_i, t_i) \in D_H$, $H$ contains an $s_i$-$r$ path \emph{and} an
$r$-$t_i$ path.\footnote{This definition of junction tree does not necessarily
correspond to a tree in a digraph;
the terminology originated from the undirected setting.}
The \emph{density} of $H$ is $c(H)/|D_H|$. The proof of \Cref{thm:main} proceeds
in two steps; first, we show that there exists a junction tree of low density,
and second, we provide an algorithm to efficiently find a low-density junction
tree given that one exists.

The key technical contribution of this paper is the first step:
showing that any instance of planar-DSF contains a low-density junction tree.
We prove the following Theorem in \Cref{sec:existence}:

\begin{restatable}{theorem}{existence}
\label{thm:existence-main}
  Given an instance $(G, D)$ of planar-DSF,
  there exists a junction tree of density $O(\log^2 k) \opt/k$ in $G$
  where $k = |D|$ and $\opt$ is the cost of an optimum solution for $(G, D)$.
\end{restatable}

We prove the preceding theorem by considering an optimum solution
$E^* \subseteq E$; we find several junction trees in $E^*$ that
are mostly disjoint and, in total, cover a large fraction of
terminal pairs.
\begin{remark}\label{rem:dsf_integrality_gap}
  We note that this proof strategy shows that
  there exists a low-density junction tree with respect to the optimal
  \emph{integral} solution. It is an interesting open problem to prove
  a poly-logarithmic factor upper bound on the integrality
  gap of the natural cut-based LP relaxation for planar-DSF.
\end{remark}
We employ
two tools developed by Thorup \cite{thorup_compact_2004} on directed
planar graphs in his work on reachability and approximate shortest
path oracles.  The first is a ``layering'' of a directed graph such
that every path is contained in at most two consecutive layers, and
each layer contains some nice tree-like structure. This allows us to
restrict our attention to two layers at a time. We remark that a
similar layering approach was recently used by
\cite{kawarabayashi_embeddings_2021} to obtain improved upper bounds
on the multicommodity flow-cut gap in directed planar graphs; this was
partly the inspiration for this work. The
second tool is a ``separator'' theorem
(essentially proved in \cite{LiptonTarjan79}, but given explicitly 
in \cite{thorup_compact_2004}), which states that every planar
graph contains three short paths whose removal results in connected
components each containing at most half the number of vertices.  We
use this shortest-path separator to devise a recursive approach
similar to that of \cite{thorup_compact_2004}.  We then restrict
attention to one level of recursion in which many $s_i$-$t_i$ paths
pass through the separator, and show that in this case, we can use the
nodes on the separator as roots for low-density junction trees.

For the second step, we prove the following theorem in \Cref{sec:algo}:

\begin{restatable}{theorem}{algo}
\label{thm:algo-main}
  Given an instance $(G, D)$ of planar-DSF, there exists an efficient algorithm
  to obtain a junction tree of $G$ of density at most $O(\log^3 k)$ times
  the optimal junction tree density in $G$.
\end{restatable}

We show that any approximation algorithm for planar-DST with respect
to the optimal \emph{fractional} solution to an LP relaxation can be
used to derive an approximation algorithm for finding the min-density
junction tree. This uses a standard bucketing and scaling argument,
initially given in the context of junction trees for buy-at-bulk
network design \cite{chekuri_approximation_2010}.  This approach
crucially relies on the fact that there exists a good approximation
algorithm for planar-DST with respect to the natural LP relaxation. Finding
junction trees without using the LP is not as straight forward.  It is
also not enough to be able to solve the min-density Directed Steiner
Tree problem; there exist algorithms to do so in planar graphs via
purely combinatorial techniques \cite{chekuri_directed_2024},
however this approach does
not extend to finding a good density junction tree due to the
additional requirement that for each $(s_i, t_i) \in D_H$, we
need to connect \emph{both} $s_i$ and $t_i$ to $r$.

\subsection{Related Work}
\label{sec:rel-work}

\mypara{DST in general digraphs:} Directed Steiner Tree was first
studied in approximation by Zelikovsky \cite{Zelikovsky97}, who
obtained an $O(k^\eps)$-approximation for any fixed $\eps >
0$. Charikar et al.  \cite{charikar_approximation_1999} built on ideas
from \cite{Zelikovsky97} to devise an $O(\log^3 k)$-approximation
in quasi-polynomial time. This was later improved to
$O(\log^2 k/ \log \log k)$ in quasi-polynomial time, by Grandoni et
al.  \cite{GrandoniLL22} who used an LP-based approach, and by Ghuge
and Nagarajan \cite{GhugeN22} who used a recursive greedy approach
building on ideas in \cite{ChekuriP05}. In terms of hardness, it is
not difficult to see that DST generalizes Set Cover and is therefore
hard to approximate to a factor $(1 - \eps) \log k$ \cite{Feige98}; in
fact, it is hard to approximate to a factor
$\Omega(\log^{2 - \eps}(k))$ unless NP has randomized quasipolynomial
time algorithms \cite{HalperinK03}. Grandoni et
al. \cite{GrandoniLL22} recently showed that even with
quasi-polynomial time algorithms, DST is not approximable within a
factor $\Omega(\log^2 k/\log \log k)$ unless the Projection Games
Conjecture fails or $\NP \subseteq \ZPTIME(2^{n^{\delta}})$ for some
$\delta \in (0,1)$.  DST and algorithmic ideas for it are closely
related to those for Group Steiner tree (GST) and Polymatroid Steiner
tree (PST). We refer the reader to some relevant papers
\cite{GargKR00,ZosinK02,chekuri_set_2011,Calinescu_Zelikovsky_2005,chekuri_directed_2024}
for more details.

\mypara{DSF and Junction Schemes:} The first nontrivial approximation
for Directed Steiner Forest was an $\tilde O(k^{2/3})$-approximation
given by Charikar et al. \cite{charikar_approximation_1999}.
This follows a similar iterative
density-based procedure as the junction tree approach; however, they
restrict to trees of a much simpler structure.
This approximation ratio was subsequently improved to $O(k^{\frac 1 2 +\eps})$
\cite{chekuri_set_2011}; \cite{chekuri_set_2011} showed that given an
instance $(G, D)$ of DSF, there exists a junction tree of density at
most $O(k^{1/2})$ times the optimum. They then provide an algorithm
to find a low-density junction tree via height reduction and
Group Steiner Tree rounding. 
DSF has improved approximation ratios when $k$ is large. 
\cite{feldman_improved_2012} obtained an 
$O(n^{\eps} \cdot \min(n^{4/5},m^{2/3}))$-approximation using a junction-based 
approach. This analysis was refined by \cite{BermanBKRY13} using ideas developed for 
finding good directed spanners, giving an improved approximation ratio of 
$O(n^{2/3 + \eps})$. DSF with \emph{uniform} edge costs admits an
$O(n^{3/5 + \eps})$-approximation \cite{CDKL20}.

DSF and DST have also been considered from a parameterized complexity
perspective. DST is fixed parameter tractable parameterized by the
number of terminals \cite{dreyfus_steiner_1971}. On the other hand,
DSF is $W[1]$-hard \cite{guo_parameterized_2011}; however, it is
polynomial time solvable if the number of terminals $k$ is constant
\cite{feldman_directed_2006,feldmann_complexity_2023}.

\mypara{Undirected Graphs:} Steiner Tree admits a simple
2-approximation by taking a minimum spanning tree on the terminal set.
There has been a long line of work improving this approximation factor
using greedy techniques
\cite{zelikovsky_116-approximation_1993,berman_improved_1994,karpinski_new_1997,hougardy_1598_1999},
culminating in a $\left(1 + \frac {\ln 3}2\right)$-approximation given by Robins
and Zelikovsky \cite{robins_tighter_2005}. This remained the best
known approximation ratio for several years, until Byrka et
al. developed an LP-based $(\ln 4 + \eps)$-approximation
\cite{ByrkaGRS13,GoemansORZ12}.  The Steiner Tree problem is APX-hard
to approximate; in fact, there is no approximation factor better than
$\frac {96}{95}$ unless $\PTIME = \NP$ \cite{ChlebikC08}.  The Steiner
Forest problem in undirected graphs admits a $2$-approximation via
primal-dual techniques \cite{agrawal_when_1995,goemans_general_1995}
and iterated rounding \cite{jain_factor_2001}. The node weighted
versions of Steiner Tree and Steiner Forest admit an 
$O(\log k)$-approximation where $k$ is the number of 
terminals \cite{klein_nearly_1995}, and
further this ratio is asymptotically tight via a reduction from Set
Cover.

\mypara{Planar and Minor-Free Graphs:}
Improved approximation ratios have been obtained for several problems
in special classes of graphs, such as planar and minor-free graphs.
We first discuss undirected graphs.
In planar graphs, Steiner Tree admits a PTAS \cite{BorradaileKM09};
this was later extended
to a PTAS for Steiner Forest in graphs of bounded genus
\cite{bateni_approximation_2011}.
Recently, \cite{cohen-addad_bypassing_2022} obtained a QPTAS for Steiner Tree in
minor-free graphs. Furthermore, although the node-weighted variant
of Steiner Tree captures Set Cover in general graphs,
there exists a constant factor approximation in
planar graphs, and more generally, in any proper minor-closed
graph family \cite{demaine_node-weighted_2009}. In directed graphs,
along with the recent results discussed above, Friggstad and Mousavi
obtained a constant-factor approximation for DST in minor-free graphs
in the setting where the input graph is \emph{quasi-bipartite}
\cite{FriggstadM21}.

\subsection{Definitions and Notation}

For a directed graph $G$, we let $V(G)$ and $E(G)$ denote the 
vertex and edge sets of $G$ respectively. For $E' \subseteq E$, we let $V(E')$ denote
the set of vertices in the graph induced by $E'$. For a subset
$S \subseteq V$, we let $\delta^+(S)$ denote the set of all edges
$(u,v)$ with $u \in S$, $v \notin S$, and we let $\delta^-(S)$ denote
the set of all edges $(u,v)$ with $u \notin S$, $v \in S$.
We will sometimes consider the \emph{undirected} version of $G$;
this is the underlying undirected graph obtained by ignoring orientations
of edges in $E$.

For any directed path (dipath) $P \subseteq G$, and for any
$u, w \in P$, we write $u <_P w$ if $u$ appears before $w$ in $P$.  We
define $>_P$, $\leq_P$, and $\geq_{P}$ similarly. For  $u, w \in P$
with $u \leq_P w$, we let $P[u, w]$ denote the subpath of $P$ from $u$
to $w$.  We denote the length of a path $P$, which is the
number of edges in $P$, by $|P|$.  For any $u,v,w \in G$, if $P'$ is a $u$-$v$
path and $P''$ is a $v$-$w$ path, we let $P' \circ P''$ denote the
concatenation of $P'$ and $P''$.  We will sometimes abuse notation and
conflate path with dipath when it is clear from context. Unless explicitly 
stated, we 
do not distinguish between paths and walks, since we are only concerned with
reachability. 

Given an instance $(G, D)$ of planar-DSF, we let $\opt$ denote the
value of an optimal solution.

\begin{definition}[Junction Tree]
\label{def:junction}
A \emph{junction tree} on terminal pairs
$D_H \subseteq D$ is a subgraph $H \subseteq G$ with a root $r$ such that for
every terminal pair $(s_i, t_i) \in D_H$, $H$ contains an $s_i$-$r$ path and an
$r$-$t_i$ path. The \emph{density} of a junction tree is the ratio of its
cost $c(H)$ to the number of terminal pairs $|D_H|$. We say a terminal
pair $(s_i, t_i) \in D$ is \emph{covered} by $H$ if $(s_i, t_i) \in D_H$; that is, there
exists an $s_i$-$t_i$ walk in $H$ containing $r$.
\end{definition}

For ease of notation, when considering subsets of terminal pairs $D' \subseteq D$, 
we sometimes write $i \in D'$ to mean $(s_i, t_i) \in D'$. 

\section{Existence of a good junction tree}
\label{sec:existence}

This section proves \Cref{thm:existence-main}, restated below:

\existence*

\begin{definition}
  A \emph{2-layered spanning tree} of a digraph $G$ is a rooted tree that is a 
  spanning tree of the undirected version of $G$ such that any path from the root 
  to a leaf is the concatenation of at most $2$ dipaths of $G$.
  A \emph{2-layered digraph} is a digraph that has a 2-layered spanning tree. 
  The \emph{root} of a 2-layered digraph is the root of its 2-layered spanning tree. 
\end{definition} 

\begin{remark}
  Note that a two-layered digraph may have additional edges aside from the spanning tree;
  we do not pose any restrictions on the directions of these edges. See Figure 
  \ref{fig:two-layered-tree} for an example. 
\end{remark}

\begin{figure}
  \centering
  \includegraphics[width=0.5\linewidth]{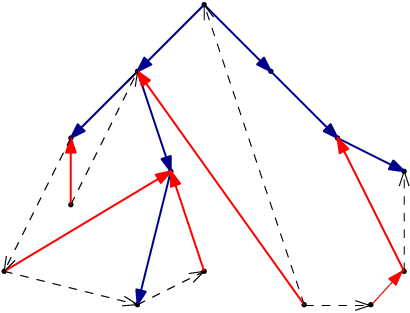}
  \caption{Example of a two-layered digraph. Bolded edges form the 
  two-layered spanning tree; remaining edges in the graph are dashed. 
  The two dipaths for each root to leaf path are denoted by blue and red edges:
  the first dipath away from the root given in blue and the second towards 
  the root in red.}
  \label{fig:two-layered-tree}
\end{figure}

The proof of \Cref{thm:existence-main} consists of three stages. 
First, in \Cref{sec:2-layered}, we use a decomposition  
given by Thorup \cite{thorup_compact_2004} of a directed graph into several 
2-layered digraphs while preserving planarity. 
Using this decomposition, we show that it suffices to consider cases where the 
optimal solution is a 2-layered digraph; thus reducing proving 
\Cref{thm:existence-main} to proving \Cref{lem:two-layered}:

\begin{restatable}{lemma}{twolayered}
\label{lem:two-layered}
  Let $(G, D)$ be an instance of planar-DSF. Suppose there 
  exists a feasible solution $E^* \subseteq E(G)$ such that $G^* := (V(E^*), E^*)$ is a 
  2-layered digraph. Let $r$ denote the root of $G^*$.
  Suppose that for each $(s_i, t_i) \in D$, 
  there exists an $s_i$-$t_i$ path in $G^* \setminus \{r\}$. 
  Then there exists a junction tree $H \subseteq G^* \setminus \{r\}$
  of density at most $O(\log^2 k)c(E^*)/k$. 
\end{restatable} 

In \Cref{sec:separator}, we use a recursive procedure built on 
a separator lemma on planar digraphs \cite{thorup_compact_2004} to reduce proving 
\Cref{lem:two-layered} to \Cref{lem:one-path}:
\begin{restatable}{lemma}{onepath}
\label{lem:one-path}
  Let $(G, D)$ be an instance of planar-DSF. Suppose there 
  exists a feasible solution $E^* \subseteq E(G)$ that contains a dipath 
  $P \subseteq E^*$ with the following property: every terminal pair 
  $(s_i, t_i) \in D$ has a dipath $P_i \subseteq E^*$ from $s_i$ to $t_i$
  such that $V(P) \cap V(P_i) 
  \neq \emptyset$. Then $E^*$ contains a junction tree of density 
  at most $O(\log k)c(E^*)/k$.
\end{restatable}

We conclude by proving \Cref{lem:one-path} in \Cref{sec:one-path}.

\subsection{Reduction to 2-Layered Digraphs}
\label{sec:2-layered}

In this section, we show that \Cref{lem:two-layered} suffices to prove
\Cref{thm:existence-main}, thus reducing to the case where the optimal
solution is a 2-layered digraph.  Let $(G, D)$ be an instance of
planar-DSF, and let $G^* = (V^*, E^*)$ be an optimal feasible solution
of cost $\opt$.  We assume without loss of generality that $E^*$
induces a weakly connected graph; if not, we apply this decomposition on each weakly
connected component separately.  We use a decomposition of digraphs
given by Thorup \cite{thorup_compact_2004}.  We include the details
and proofs here for the sake of completeness, and to highlight some
additional properties that we need.  Let $v_0 \in V^*$ be an arbitrary
node in $G^*$.  We let $L_0$ be the set of
all nodes in $V^*$ that are reachable from $v_0$ in $G^*$.  Then, we
define alternating ``layers'' as follows:
\begin{align*}
  L_j = \begin{cases}
    \{v \in V^* \setminus \cup_{j' < j} L_{j'}: v \text{ can reach } L_{j-1} \text{ in } G^*\} &j \text{ is odd} \\
    \{v \in V^* \setminus \cup_{j' < j} L_{j'}: v \text{ is reachable from } L_{j-1} \text{ in } G^*\} &j \text{ is even}
  \end{cases}.
\end{align*}

\begin{figure}
  \centering
  \includegraphics[width=0.7\linewidth]{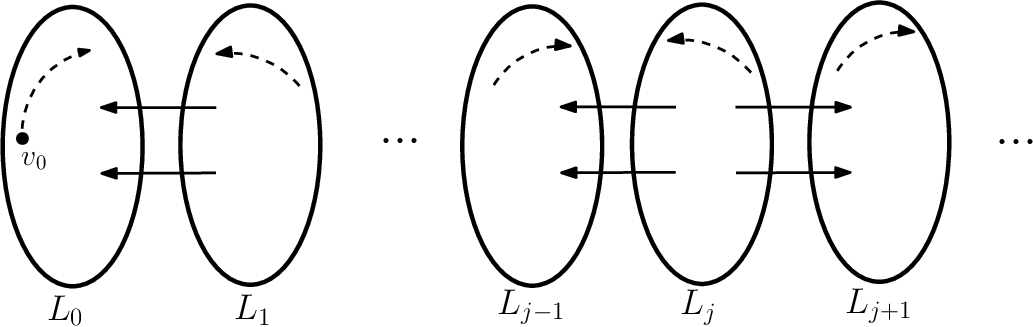}
  \caption{Layers constructed from $G^*$. Dotted lines represent edges inside
  each layer, while solid lines represent edges between layers. In this example,
  $j$ is odd.}
  \label{fig:layers}
\end{figure}

We continue this process until all vertices in $V^*$ are covered by a layer;
see \Cref{fig:layers}.
Let $\ell$ denote the index of the last layer. For $j \in \{0, \dots, \ell-1\}$,
we define $G_j$ to be the graph obtained from $G^*$ by
deleting all nodes in $\cup_{i > j+1} L_i$
and contracting all nodes in $\cup_{i < j} L_j$. We call this
contracted node the \emph{root} $r_j$ of $G_j$. It is clear from construction
that each $G_j$ is a 2-layered digraph.
Furthermore, each $G_j$ is a minor of
$G^*$ and is thus planar.\footnote{A graph $H$ is a \emph{minor} of $G$ if it
can be obtained from $G$ by deleting and/or contracting edges of $G$. It is
easy to see that if $G$ is planar, then any minor of $G$ is planar as well.}

\begin{claim}
\label{claim:layering_cost}
  The total cost $\sum_{j = 0}^{\ell-1} c(E(G_j)) \leq 2c(E^*)$.
\end{claim}
\begin{proof}
  We show that each edge of $E^*$ appears in at most two of the graphs
  from $G_0,\ldots, G_{\ell-1}$.  Since $E(G_j) \subseteq E^*$ for all
  $j$, the claim follows.  Let $(u,v) \in E^*$. If $u, v$ are in the
  same layer $L_j$, then $(u,v)$ is only in $G_j$ and $G_{j-1}$; all
  other graphs $G_{j'}$ either contract (when $j' > j$) or delete
  (when $j' < j-1$) $L_j$.  If $u,v$ are in distinct layers, they must
  be in adjacent layers $L_j$ and $L_{j+1}$. For $j' > j+1$, both
  $L_j$ and $L_{j+1}$ are contracted into the root, thus
  $(u,v) \notin G_{j'}$.  For $j' < j$, $L_{j+1}$ is deleted, thus
  once again $(u,v) \notin G_{j'}$.  Therefore the edge $(u,v)$ can
  only appear $G_j$ and/or $G_{j+1}$.
\end{proof}

\begin{claim}
\label{claim:layering_pairs_covered}
  For each pair $(s_i,t_i) \in D$, there exists some $j \in \{0, \dots, \ell-1\}$
  such that $L_j \cup L_{j+1}$ contains an $s_i$-$t_i$ path.
\end{claim}
\begin{proof}
  Let $P_i$ be an $s_i$-$t_i$ path in $E^*$; such a path must exist by
  feasibility of $E^*$. Let $j$ be the minimum index such that $L_j$
  intersects $P_i$, and let $v$ be a node in
  $L_j \cap P_i$.

  If $j$ is even, any node reachable
  from $L_j$ must be contained in $\cup_{j' \leq j} L_{j'}$; thus 
  $P_i[v, t_i] \subseteq \cup_{j' \leq j} L_{j'}$.  
  By definition, $L_{j+1}$ contains all nodes in $G^* \setminus \cup_{j' \leq j} L_{j'}$
  that can reach $L_j$;
  thus $P_i[s_i,v]$ must be contained in $\cup_{j' \leq j+1} L_{j'}$.
  
  Otherwise, if $j$ is odd, any node that can reach
  $L_j$ must be contained in $\cup_{j' \leq j} L_{j'}$, so $P_i[s_i, v] 
  \subseteq \cup_{j' \leq j} L_{j'}$. $L_{j+1}$ contains all nodes
  in $G^* \setminus \cup_{j' \leq j} L_{j'}$ 
  reachable from $L_j$, so $P_i[v, t_i] \subseteq \cup_{j' \leq j+1} L_{j'}$.

  In either case, $P_i \subseteq \cup_{j' \leq j+1} L_{j'}$. Since $j$ is the minimum 
  index that intersects $P_i$, $P_i \subseteq L_j \cup L_{j+1}$ as desired.
\end{proof}

\begin{proof}[Reduction from \Cref{thm:existence-main} to \Cref{lem:two-layered}]
  We partition the demand pairs $D$ into $D_0, \dots, D_{j-1}$,
  where $(s_i, t_i) \in D_j$ if $L_j \cup L_{j+1}$ contains an 
  $s_i$-$t_i$ path;
  if there are multiple such $j$ we choose one arbitrarily.
  Note that all terminal pairs are covered by this partition by
  \Cref{claim:layering_pairs_covered}.

  Since $\sum_{j = 0}^{\ell-1} c(E(G_j)) \leq 2c(E^*)$ and $D_0, \dots, D_{j-1}$
  form a complete partition of $D$, there must be some $j \in \{0, \dots, \ell-1\}$
  such that $c(E(G_j))/|D_j| \leq 2c(E^*)/|D|$. We claim that $(G_j, D_j)$ 
  satisfies the conditions of \Cref{lem:two-layered}: 
  $G_j$ is a planar 2-layered digraph that is a feasible solution 
  on all terminal pairs $D_j$, and for each $i \in D_j$, there is an $s_i$-$t_i$ 
  path contained in $L_j \cup L_{j+1}$, thus avoiding the root of $G_j$. 
  By \Cref{lem:two-layered}, there exists a junction tree $H$ in $G_j$ of 
  density
  \[O(\log^2 |D_j|) \frac{c(E(G_j))}{|D_j|}
  \leq O(\log^2 k) \frac{2c(E^*)}{|D|} = O(\log^2 k) \opt/k.\]

  Furthermore, since $H$ does not contain the root of $G_j$, 
  $H$ is a subgraph of $G^*$.   
\end{proof}

\subsection{Reduction from Two-Layered Digraphs to One-Path Setting}
\label{sec:separator}

In this section, we show that assuming \Cref{lem:one-path}, we can prove
\Cref{lem:two-layered}, restated below:

\twolayered*

Fix an instance $(G, D)$ of planar-DSF and a feasible solution $E^*$ satisfying 
the conditions outlined in the statement of \Cref{lem:two-layered} above. Let 
$T^* \subseteq E^*$ be a 2-layered spanning tree with root $r$. 
Given any \emph{undirected} tree $T$, we let $P_T(u,v)$ denote the unique tree path
from $u$ to $v$. We will follow a
recursive process to partition $D$ into subsets on which we build junction trees.
To do so, we use the following separator lemma on planar digraphs.

\begin{lemma}[\cite{thorup_compact_2004}]
\label{lem:thorup-planar-sep}
  Given an \emph{undirected} planar graph $G = (V, E)$ with a spanning tree $T$
  rooted at $r$ and non-negative vertex weights $w: V \to \R_{\geq 0}$,
  we can find three vertices $u_1, u_2, u_3$ such that each component of
  $G \setminus (P_T(r, u_1) \cup P_T(r, u_2) \cup P_T(r, u_3))$ has at most
  half the weight of $G$.
\end{lemma}

We define vertex weights $w(v) = 1$ if $v \in D$ and $w(v) = 0$ otherwise.
We consider the undirected version of $E^*$; that is, we ignore
all directions on $E^*$ and apply \Cref{lem:thorup-planar-sep}
on the undirected version of spanning tree $T^*$ with vertex weights $w$. 
From this, we obtain
$u_1, u_2, u_3$.
Since $T^*$ is a 2-layered spanning tree, each path
$P_{T^*}(r, u_i)$ consists of at most 2 dipaths of $E^*$. We remove the 
root $r$ and let $Q_i^1, Q_i^2$ denote the at most two dipaths of 
$P_{T^*}(r, u_i) \setminus \{r\}$. 
Let $S_0 = \cup_{i \in [3]} \{Q_i^1, Q_i^2\}$
denote this set of at most 6 dipaths; we call this a \emph{separator}.
We define $D_0 \subseteq D$ to be the set
of all terminal pairs $(s_i, t_i)$ such that $E^*$ contains an $s_i$-$t_i$ path
going through one of the dipaths in $S_0$. Equivalently, $(s_i, t_i) \in D_0$ iff
there exists an $s_i$-$t_i$ path $P_i \subseteq E^*$ such that
$V(P_i) \cap V(S_0) \neq \emptyset$.
See \Cref{fig:separator} for an example of a separator with the
corresponding set $D_0$.
We let $\calC_0$ be the set of weakly connected
components of $G \setminus (\cup_{i \in [3]} P_{T^*}(r, u_i))$; we 
drop ``weakly connected'' and simply refer to these
as ``components'' in the remainder of this section. Note that
each $C \in \calC_0$ has at most half the total number of terminals.

\begin{figure}
  \centering
  \includegraphics[width=0.7\linewidth]{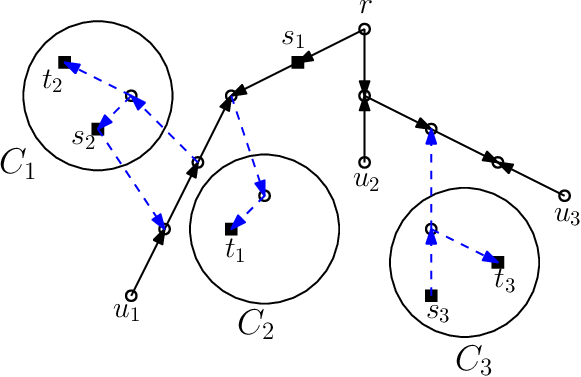}
  \caption{Example of separator and resulting weakly connected components.
  Solid black lines
  denote edges in the separator $S_0$, while dashed blue lines represent
  edges between components and the separator. Terminals are labeled and denoted 
  with boxes. In this example, $D_0 = \{(s_1, t_1), (s_2, t_2)\}$
  since there exists an $s_1$-$t_1$ and an $s_2$-$t_2$ path through
  the separator. Notice that $(s_2, t_2) \in D_0$ even though $s_2$ and $t_2$ remain
  in the same component $C_1$.}
  \label{fig:separator}
\end{figure}

We recurse on each component $C \in \calC_0$ as follows:
we contract $S_0$ into $r$ and recurse on the
sub-instance consisting of $C$ and the new contracted root $r$. 
It is not difficult to see that this new sub-instance is a 2-layered 
digraph and thus contains a 2-layered spanning tree $T^*_C$. 
We repeat the same process as above, applying \Cref{lem:thorup-planar-sep}
with $T^*_C$, and weights the same as before for all nodes in $C$ and 
$w(r) = 0$. We obtain three nodes $u_1', u_2', u_3' \in C$. Once again, 
we ignore
$r$ when considering the dipaths. We define $S_1^C$ to be the set of
at most 6 dipaths in
$\cup_{i \in [3]} (P_T(r, u_i') \setminus \{r\})$, and let $D_1^C$
be the set of all $(s_i, t_i) \in D$ with $s_i, t_i \in C$ such that
there exists an $s_i$-$t_i$ path in $C$ with a non-empty intersection
with a dipath in $S_1^C$.

\begin{remark}
  In the recursive step, we choose to contract the separator into
  the root to maintain the property that each recursive call still corresponds to
  a 2-layered digraph. It is important to remove the
  contracted root $r$ from the dipaths of $S_i$, $i > 0$ to ensure
  that all nodes in the separator are nodes in $G$ and all
  separators are disjoint.
  We remove the root $r$ from the dipaths of $S_0$ to ensure that 
  $H$ does not contain the root of $E^*$, in order to satisfy the lemma 
  statement. 
\end{remark}

We continue this recursive process until
each component has at most one terminal.
Since the number of terminals halve
at each step, the total recursion depth is at most $\lceil \log 2k
\rceil = \lceil \log k \rceil + 1$.
For ease of notation, for $j \geq 1$
we denote by $S_j := \cup_{C \in \calC_{j-1}}
S_j^C$ the set of all dipaths constructed in the $j$th level of recursion
and let $D_j := \cup_{C \in \calC_{j-1}} D_j^C$.

\begin{claim}
\label{claim:all-terminals-sep}
  $D \subseteq \cup_{j \in 0}^{\lceil \log k \rceil + 1} D_j$.
\end{claim}
\begin{proof}
  Fix $(s_i, t_i) \in D$, and let $P_i$ be an $s_i$-$t_i$ path in $E^* \setminus \{r\}$. 
  Let $j$ be
  the first recursive level such that $P_i$ intersects $S_j$; such a level
  must exist since by the last step of recursion, $s_i$ and $t_i$ are in different
  components. Let $C$ be the component such that $P_i$ intersects $S_j^C$.
  Then $P_i$ must be fully contained in $C$; else $P_i$ would have intersected
  a separator at an earlier level. Thus $(s_i, t_i) \in D_j^C \in D_j$.
\end{proof}

\begin{corollary}
\label{cor:good-recursion-level}
  There exists a recursion level $j^* \in \{0, \dots, \lceil \log k \rceil + 1\}$
  such that $|D_{j*}| \geq \frac k {\lceil \log k \rceil + 2}$.
\end{corollary}

\Cref{cor:good-recursion-level} allows us to focus on one recursion layer that covers a large number of
terminal pairs, and use the planar separators $S_{j^*}$ to reduce to the one
path case.

\begin{proof}[Reduction from \Cref{lem:two-layered} to \Cref{lem:one-path}]
  Let $j^*$ be the recursion level given by
  \Cref{cor:good-recursion-level} such that
  $|D_{j*}| \geq \frac k {\lceil \log k \rceil + 2}$.
  Recall that we define $\calC_{j^*}$ to be the set of all components at level
  $j^*$.
  Note that all components $C \in \calC_j^*$ are disjoint; therefore,
  $\sum_{C \in \calC_{j^*}} c(E(C)) \leq c(E^*)$. Furthermore, since
  $(s_i, t_i) \in D_{j^*}^C$ implies that $s_i, t_i \in C$, $D_{j^*}^C$ form
  a partition of $D_{j^*}$.
  Thus there must be one component $C \in \calC_{j^*}$
  such that $c(E(C))/|D_{j^*}^C| \leq c(E^*)/|D_{j^*}|$. Fix
  this component $C$.

  By construction, for all $(s_i, t_i) \in D_{j^*}^C$ there is an $s_i$-$t_i$
  path intersecting $S_{j^*}^C$ that is fully contained in 
  $C \setminus \{r\}$.
  Since $S_{j^*}^C$ consists of
  at most 6 dipaths, there must be at least one dipath, which we call
  $Q_{j^*}^C$, such that at least $\frac 1 6$ of the terminal pairs in $D_{j^*}^C$
  have paths that intersect $Q_{j^*}^C$; we denote this subset of terminal pairs by
  $D^*$.
  We apply \Cref{lem:one-path} on $(C \setminus \{r\}, D^*)$ to obtain
  a junction tree $H$ of density at most
  \[O(\log |D^*|) \frac{c(C)}{|D^*|} \leq
  O(\log k) \frac{6c(C)}{|D_{j^*}^C|} \leq 
  O(\log k) \frac{6c(E^*)}{|D_{j^*}|} \leq O(\log^2 k) c(E^*)/k.\]
\end{proof}

\subsection{One-Path Setting}
\label{sec:one-path}

In this section we prove \Cref{lem:one-path}, restated below:

\onepath*

Fix an instance $(G, D)$ of planar-DSF and a solution $E^*$ with
dipath $P \subseteq E^*$ satisfying the conditions outlined in the
statement of \Cref{lem:one-path} above. We will sometimes overload
notation and write $P$ as $V(P)$.  We label the vertices on $P$ as
$v_0, \dots, v_{|P|}$.  For each terminal pair $s_i$-$t_i$, let $a_i$
denote the first node in $P$ that $s_i$ can reach, and let $b_i$
denote the last node in $P$ that can reach $t_i$ (here, reachability
is defined using edges in $E^*$).  By the condition in
\Cref{lem:one-path}, $a_i \leq_P b_i$ for all $i \in [k]$; else no
$s_i$-$t_i$ path could intersect $P$. We let $I_i$ denote the interval
$P[a_i,b_i]$.  We let $P_{s_i}$ denote the path in $E^*$ from $s_i$ to
$a_i$ and let $P_{t_i}$ denote the path in $E^*$ from $b_i$ to
$t_i$. See \Cref{fig:one-path} for an example.

\begin{figure}
  \centering
  \includegraphics[width=0.7\linewidth]{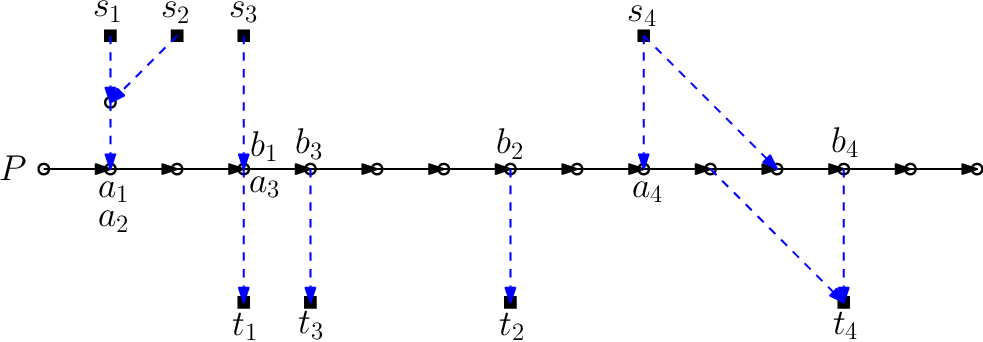}
  \caption{The path $P$ is given with solid black lines. Blue dashed
  lines represent the paths between terminals and $P$. Note that terminals
  can have multiple paths to/from $P$, as shown by $s_4$/$t_4$. In this example,
  terminal pairs $1,2,3$ all have mutually overlapping intervals and thus form
  a junction tree rooted at the vertex $b_1 = a_3$.}
  \label{fig:one-path}
\end{figure}

We start with a simple observation regarding these intervals
and their relation to junction trees; we show that if there exists a set of
intervals which all overlap at a common vertex, then we can form a junction tree
on the corresponding terminal pairs.

\begin{claim}
\label{claim:one-path-overlap}
  Let $D' \subseteq D$ such that $\cap_{i \in D'} I_i \neq \emptyset$, i.e.
  all intervals overlap.
  Let $a_{\text{start}} = \min_{i \in D'} a_i$ and $b_{\text{end}}
  = \max_{i \in D'} b_i$,
  where min and max are taken with respect to $\leq_P$. Then
  $H = P[a_{\text{start}},b_{\text{end}}] \cup \bigcup_{i \in D'}
  (P_{s_i} \cup P_{t_i})$ is a valid junction tree on $D'$.
\end{claim}
\begin{proof}
  Let $v$ be some element in $\cap_{i \in D'} I_i$; this will be the root
  of the junction tree $H$. It suffices to show that for all $(s_i, t_i) \in D'$,
  $s_i$ can reach $v$ and $v$ can reach $t_i$ in $H$.
  Let $(s_i, t_i) \in D'$. By definition of $a_{\text{start}}$ and $b_{\text{end}}$,
  and since $v \in I_i$, we have that
  $a_{\text{start}} \leq_P a_i \leq_P v \leq_P b_i \leq_P b_{\text{end}}$.
  Therefore $P[a_i, v]$ and $P[v, b_i]$ are contained in $H$. Thus the
  $s_i$-$v$ path $P_{s_i} \circ P[a_i,v]$ is contained in $H$, as is the
  $v$-$t_i$ path $P[v, b_i] \circ P_{t_i}$, as desired.
\end{proof}

\Cref{claim:one-path-overlap} provides a natural way to obtain junction
trees in $E^*$: we partition $D$ into groups such that in each group,
all corresponding intervals overlap at a common vertex, and then form the
junction trees accordingly.
To partition $D$, we first separate terminal pairs based on their interval lengths;
recall that path lengths are defined in terms of number of edges, so the
length of the interval $I_i$ is the number of edges from $a_i$ to $b_i$ in $P$.
We let $D_0$ denote the set of all $(s_i, t_i) \in D$ such that $a_i = b_i$; these
correspond to $0$-length intervals.
For $j \in \{1, \dots, \log |P|+1\}$, let
$D_j = \{(s_i, t_i): |I_i| \in [2^{j-1}, 2^j)\}$. For $v \in P$,
we let $D_j^v \subseteq D_j$ be the set of all $(s_i, t_i) \in D_j$ such that
$v \in I_i$. We construct the set of groups
\[\calG = \{D_0^v: \exists i \in D_0 \text{ s.t. } a_i = b_i = v\} \cup
\bigcup_{j \in [\log|P|+1]} \{D_j^{v_{\ell}}: \ell \text{ is
a multiple of } 2^{j-1}\}.\]
Note that for each group $D_j^v$, $v \in \cap_{i \in D_{j}^v} I_i$.
Therefore, each group $D_j^v$ is associated with a
junction tree $H_j^v$ with root $v$ as given by \Cref{claim:one-path-overlap}.
We let $\calH$ denote the set of all such junction trees.

\begin{claim}
\label{claim:group-covers-intervals}
  Every $(s_i, t_i) \in D$ is in some group in $\calG$.
\end{claim}
\begin{proof}
  Fix $(s_i, t_i) \in D$. If $a_i = b_i$, then $(s_i, t_i) \in D_0$, so 
  $(s_i, t_i) \in D_0^{a_i} \in \calG$.
  Else, $|I_i| \in \{1, \dots, |P|\}$, so $\exists j \in
  \{1, \dots, \log|P| + 1\}$ such that $(s_i, t_i) \in D_j$. Let $\ell$ be the first
  multiple of $2^{j-1}$ such that
  $v_{\ell} \geq_P a_i$. Then $P[a_i, v_{\ell}] \leq 2^{j-1}$. Since $(s_i, t_i) \in D_j$,
  $|I_i| \geq 2^{j-1}$; thus $v_{\ell} \leq_P b_i$. Therefore
  $v_{\ell} \in I_i$, so $(s_i, t_i) \in D_j^{v_\ell} \in \calG$.
\end{proof}

We will show that the junction trees in $\calH$ have, on average, low
density. To do so,
one must ensure that each edge $e \in E^*$ only appears in
$O(\log k)$ junction trees to maintain the cost bound. A technical difficulty
is reasoning about the edges of $E^* \setminus P$,
since the paths between terminals
and the path $P$ may intersect and share edges. The following key observation
provides some structure on these paths with respect to the intervals:

\begin{claim}
  \label{claim:one-path-disjoint}
  For any $i, i' \in [k]$, if
  $P_{s_i} \cap P_{s_{i'}} \neq \emptyset$, then $a_i = a_{i'}$.
  Similarly, $P_{t_i} \cap P_{t_{i'}} \neq \emptyset$, then $b_i = b_{i'}$.
\end{claim}
\begin{proof}
  Consider $i, i' \in [k]$, and suppose without loss of
  generality that $a_i \leq_P a_{i'}$. Let $v \in P_{s_i} \cap P_{s_{i'}}$.
  Then $P_{s_{i'}}[s_{i'},v] \circ P_{s_i}[v,a_i]$ is a path from
  $s_{i'}$ to $a_i$ in $E^*$. Since we defined $a_{i'}$ as the earliest
  point that $s_{i'}$ can reach on $P$, it must be the case that
  $a_{i'} \leq_P a_i$, so $a_i = a_{i'}$.
  An analogous argument shows that for any $i, i' \in [k]$,
  if $P_{t_i} \cap P_{t_{i'}} \neq \emptyset$, then $b_i = b_{i'}$.
\end{proof}

\begin{claim}
\label{claim:one-path-few-groups}
  Each node in $P$ appears in at most $5\log|P|+6$
  junction trees in $\calH$. The same holds for each edge in $P$.
\end{claim}
\begin{proof}
  Let $u \in P$. By \Cref{claim:one-path-overlap}, for any
  $H_j^v \in \calH$, $H_j^v \cap P = P[a_{\text{start}}, b_{\text{end}}]$, where
  $a_{\text{start}}$ is the first interval start point and $b_{\text{end}}$ is the last
  interval end point of all intervals of $D_j^v$. Notice that since
  the intervals of $D_j^v$ overlap at a common vertex, $P[a_{\text{start}}, b_{\text{end}}]$ 
  is equivalent to $\cup_{i \in D_{j^v}} I_i$.

  First, consider $j = 0$. In this case, for any $v$,
  $\cup_{i \in D_0^v} I_i = \{v\}$. Thus $u \in D_0^v$ if and only if
  $u = v$, so $u$ is in at most one group when $j = 0$.

  Next, fix $j \geq 1$, and consider some $v_{\ell}$ such that
  $D_j^{v_\ell} \in G$. Note that $(s_i, t_i) \in D_j^{v_\ell}$ implies that
  $|I_i| < 2^j$ and $v_\ell \in I_i$. Therefore, it must be the case
  that $a_i >_P v_{\ell - 2^j}$ and $b_i <_P v_{\ell + 2^j}$.
  Thus for any $D_j^{v_\ell} \in \calG$,
  $\cup_{i \in D_j^{v_\ell}} I_i \subseteq P[v_{\ell-2^j},v_{\ell+2^j}]$.
  Therefore if $u \in D_j^{v_\ell}$, then $v_{\ell}$ has to be within
  $2^j$ edges of $u$. Since $\calG$ only contains $D_j^{v_\ell}$ for
  $\ell$ a multiple of $2^{j-1}$, there are at most $5$ values of
  $\ell$ that are multiples of $2^{j-1}$ such that $v_{\ell}$ can either
  reach or be reached by $u$ within $2^j$ edges. Therefore, $u$ is in at most
  $5$ groups for any fixed $j$. Summing over all $j = 1, \dots, \log|P|+1$
  gives the desired bound.

  Each edge $e \in P$ is only in a junction tree $H$ if both its endpoints
  are also in $H$. Thus the same upper bound holds for each edge in $P$.
\end{proof}

\begin{claim}
\label{claim:outside-path-few-groups}
  Each node in $V(E^*) \setminus P$ appears in at most $10\log|P| + 12$
  junction trees in $\calH$. The same holds for each edge in $E^* \setminus P$.
\end{claim}
\begin{proof}
  Fix $u \in V(E^*) \setminus P$. If $u$ is in any junction tree $H$,
  it must be in some $P_{s_i}$ and/or some $P_{t_i}$; it may be in many
  such paths for various terminal pairs.
  Let $D_s = \{i: u \in P_{s_i}\}$ and $D_t = \{i: u \in P_{t_i}\}$.

  By \Cref{claim:one-path-disjoint},
  there exists some node $a \in P$ such that for all $i \in D_s$,
  $a_i = a$. Similarly, there exists some $b \in P$ such that
  for all $i \in D_t$, $b_i = b$. By construction of
  junction trees in \Cref{claim:one-path-overlap}, for any
  $H \in \calH$, $u \in H$ only if $a \in H$ or $b \in H$.
  By \Cref{claim:one-path-few-groups}, $a$ and $b$ are each
  in at most $5 \log |P| + 6$ junction trees in $\calH$. Therefore,
  $u$ is in at most $2(5 \log |P| + 6)$ junction trees in $\calH$.

  Each edge $e \in E^*$ is only in a junction tree $H$ if both its endpoints
  are also in $H$. Thus the same upper bound holds for each edge in $E^*$.
\end{proof}

We conclude the proof of the main lemma:
\begin{proof}[Proof of \Cref{lem:one-path}]
  By
  \Cref{claim:one-path-few-groups} and \Cref{claim:outside-path-few-groups},
  each edge of $E^*$ is in at most $O(\log|P|)$ junction trees, so
  $\sum_{H \in \calH} c(H) \leq O(\log |P|) c(E^*)$.
  We note that while $|P|$ could be as large as $\Theta(n)$, we can effectively assume
  $|P| \leq 2k$ as follows. First, we assume
  $E^* = P \cup_{i \in [k]} (P_{s_i} \cup P_{t_i})$; these are the
  only edges used in junction trees $\calH$ and constitutes a feasible
  solution. Then, we can ignore
  all degree-2 nodes in $P$: if $v_i \in P$ has degree 2 in $E^*$,
  we can replace the edges $e' = (v_{i-1}, v_i)$ and $e'' = (v_i, v_{i+1})$
  with an edge $e = (v_{i-1}, v_{i+1})$ of cost $c(e') + c(e'')$ without
  changing feasibility of $E^*$. The only nodes in $P$ that have degree
  greater than 2 in $E^*$ are the points $a_i, b_i$ for $i \in [k]$.
  Thus we can assume $|P| \leq 2k$, and
  $\sum_{H \in \calH} c(H) \leq O(\log k) c(E^*)$.

  By \Cref{claim:group-covers-intervals}, all terminal pairs are
  covered by at least one junction tree in $\calH$. Therefore, the
  total density of junction trees in $\calH$ is $O(\log k) c(E^*)/k$.
  An averaging argument
  shows that there must be at least
  one $H^* \in \calH$ that has density at most $O(\log k) c(E^*)/k$.
\end{proof}

\section{Finding a good junction tree}
\label{sec:algo}

In this section we show that there exists an efficient algorithm to
find an approximate min-density junction tree,
proving \Cref{thm:algo-main} restated below:
\algo*

We employ an LP-based approach. We consider a natural cut-based LP relaxation
for DST with variables $x_e \in [0,1]$ for $e \in E$ indicating whether or not
$e$ is in the solution. Here, the input is a digraph $G = (V, E)$ with root $r$
and terminals $t_{i}$, $i \in [k]$.

\begin{equation}
\label{DST-LP}
\tag{DST-LP}
\begin{aligned}
  \min\quad \sum_{e\in E}c(e)x_e& \\
  s.t.\quad \sum_{e\in \delta^+(S)} x_e&\geq 1 \quad
  \forall S \subseteq V, r \in S, \exists i \text{ s.t. } t_i \notin S \\
      x_e &\geq 0 \quad \forall\ e\in E
\end{aligned}
\end{equation}

We prove the following lemma:

\begin{lemma}
\label{lem:algo-main}
  Suppose there exists an $\alpha$-approximation for DST in planar graphs
  with respect to the optimal solution to \ref{DST-LP}. Then, given a
  planar-DSF instance $(G, D)$, there exists an efficient algorithm to
  obtain a junction tree of $G$ of density at most $O(\alpha \cdot \log k)$
  times the optimal junction tree density in $G$.
\end{lemma}

It was recently shown by \cite{chekuri_directed_2024}
that there exists an $O(\log^2 k)$ approximation
for DST in planar graphs with respect to the optimal solution to \ref{DST-LP}.
Therefore it suffices to prove \Cref{lem:algo-main} to prove
\Cref{thm:algo-main}.

Let $(G, D)$ be an instance to planar-DSF.
We start by guessing the root $r$ of the junction structure, as we can
repeat this algorithm for each $r \in V$ and choose the resulting junction
structure of minimum density. We consider the following LP relaxation for 
finding the minimum density junction tree rooted at $r$.
We follow a similar structure to that of
\ref{DST-LP}, with additional variables $y_{s_i}$ and $y_{t_i}$ for each
$i \in [k]$ to indicate whether or not $s_i$ and $t_i$ are included in the
solution. We ensure that $y_{s_i} = y_{t_i}$ so that the junction tree
includes complete pairs rather than individual terminals.
We also change the direction of flow from each $s_i$ to the root $r$.
The resulting minimum density would be $(\sum_{e \in E} c(e)x_e)/
(\sum_{i \in [k]} y_{t_i})$; we normalize $\sum_{i \in [k]} y_{t_i} = 1$.

\begin{equation}
\label{Den-LP}
\tag{Den-LP}
\begin{aligned}
  \min\quad \sum_{e\in E}c(e)x_e& \\
  s.t.\quad \sum_{e\in \delta^+(S)} x_e&\geq y_{t_i} \quad
  \forall i \in [k], \forall S \subseteq V, r \in S, t_i \notin S \\
  \sum_{e\in \delta^-(S)} x_e&\geq y_{s_i} \quad
  \forall i \in [k], \forall S \subseteq V, r \in S, s_i \notin S \\
  y_{s_i} &= y_{t_i} \quad \forall i \in [k] \\
  \sum_{i \in [k]} y_{t_i} &= 1 \\
  x_e, y_{s_i}, y_{t_i} &\geq 0 \quad \forall\ e\in E, i \in [k]\\
  &
\end{aligned}
\end{equation}

We claim that \ref{Den-LP} provides a valid lower bound for
the optimum density of a junction tree through $r$.

\begin{claim}
\label{claim:den-lp-valid}
  For any junction tree $H$ of $G$, there exists a feasible fractional
  solution $(x, y)$ to \ref{Den-LP} such that
  $\sum_{e \in E} c(e)x_e = c(H)/|D_H|$, where $D_H$ is the set
  of terminal pairs covered by $H$.
\end{claim}
\begin{proof}
  Let $H$ be any junction tree of $G$, let $D_H \subseteq D$ be the
  terminal pairs covered by $H$. Consider $(x, y)$ given by
  $x_e = 1/|D_H|$ if $e \in H$ and $0$ otherwise,
  $y_{t_i} = 1/|D_H|$ if $(s_i, t_i) \in D_H$ and $0$ otherwise, and
  $y_{s_i} = y_{t_i}$ for all $i \in [k]$. For each $(s_i, t_i) \in D_H$,
  since $H$ contains
  an $s_i$-$r$ path and an $r$-$t_i$ path,
  $x$ supports a flow of $1/|D_H|$ from $s_i$ to $r$ and $r$ to $t_i$; 
  thus the first two sets of constraints are satisfied.
  It is easy to verify that the rest of the constraints are satisfied
  and that $\sum_{e \in E} c(e)x_e = c(H)/|D_H|$.
\end{proof}
Despite the fact that the LP has exponentially many constraints, it can be
solved efficiently via a separation oracle: suppose we are given a
fractional solution $(x, y)$. The first two sets of constraints are satisfied
if for every $i \in [k]$, $x$ supports a flow of at least $y_{t_i}$
from $r$ to $t_i$ and a flow of at least $y_{s_i}$ from $s_i$ to $r$;
these can be checked via min-cut computations. There are only polynomially
many remaining constraints; thus these can be checked in polynomial
time. One can also write a compact LP via additional flow
variables. 

To find a junction tree of $G$,
We first solve \ref{Den-LP} to obtain an optimal fractional solution $(x^*, y^*)$.
For $j = 0, \dots, \log k$, we let
$D_j = \{(s_i, t_i) \in D: y_{t_i} \in (\frac 1 {2^{j+1}}, \frac 1 {2^j}]\}$.
We will show that there exists a group $\theta \in \{0, \dots, \log k\}$
for which the
total $y^*$ value is large; thus $x^*$ supports a good fraction of flow
from the root to/from $D_\theta$.

\begin{claim}
\label{claim:algo_good_group}
  There exists $\theta \in \{0, \dots, \log k\}$ such that
  $\sum_{i \in D_\theta} y_{t_i} \geq 1/(2\log k+2)$.
\end{claim}
\begin{proof}
  If $(s_i, t_i) \in D \setminus (\cup_{j =0}^{\log k} D_j)$, then
  $y_{t_i} \leq \frac 1 {2^{\log k+1}} = \frac 1 {2k}$. Therefore, the total
  $y$ value of pairs not covered by the sets $D_j$ is at most
  $\sum_{i \notin \cup_{j = 0}^{\log k} D_j} y_{t_i}
  \leq k \frac 1 {2k} = \frac 1 2$.
  Since $\sum_{i \in [k]} y_{t_i} = 1$, the total $y$ value of pairs
  covered by the sets $D_j$ is at least
  $\sum_{i \in \cup_{j = 0}^{\log k} D_j} y_{t_i} \geq \frac 1 2$.
  Since there are $\log k + 1$ disjoint groups, there is a group whose total
  $y$ value is at least $1/(2(\log k+1))$.
\end{proof}

Let $\theta$ be given by \Cref{claim:algo_good_group}. We use the
$\alpha$-approximation algorithm for DST twice: first, we consider the instance
on $G$ with terminal set $D_\theta^t = \{t_i: (s_i, t_i) \in D_\theta\}$ and obtain a directed
$r$-tree $T_t$. Second, we let $G'$ be obtained from $G$ by reversing the
direction of all edges. We apply the $\alpha$-approximation algorithm for
DST on $G'$ with terminal set $D_\theta^s = \{s_i: (s_i, t_i) \in D_\theta\}$ and obtain a
directed $r$-tree $T_s$ in $G'$. Note that $T_s$ is a directed in-tree
in $G$; therefore, $T = T_t \cup T_s$ is a valid junction tree on $G$ and
terminal pairs $D_{\theta}$.

\begin{claim}
\label{claim:algo_lp_feasible}
  $2^{\theta + 1} x^*$ is a feasible solution to \ref{DST-LP} on both of
  the following instances:
  \begin{itemize}
    \item $G$ with terminal set $D_\theta^t$,
    \item $G'$ with terminal set $D_\theta^s$.
  \end{itemize}
\end{claim}
\begin{proof}
  We first consider $G$ with terminal set $D_\theta^t$. Fix $S \subseteq V$,
  $t_i \in D_{\theta}^t$ such that $r \in S, t_i \notin S$. Since $x^*$ is a
  feasible solution to \ref{Den-LP} and $i \in D_{\theta}$,
  \[\sum_{e\in \delta_G^+(S)} 2^{\theta+1} x^*_e
  \geq 2^{\theta + 1} y_{t_i} > 2^{\theta+1}\frac 1{2^{\theta+1}} = 1.\]
  Next, we consider $G'$ with terminal set $D_{\theta}^s$, and fix $S \subseteq V$
  and $s_i \in D_{\theta}^s$ such that $r \in S, s_i \notin S$. Then, using
  the fact that $y_{s_i} = y_{t_i}$ for all $i$, we use the same argument
  as above:
  \[\sum_{e\in \delta_{G'}^+(S)} 2^{\theta+1} x^*_e
  = \sum_{e \in \delta_G^-(S)} 2^{\theta + 1} x^*_e
  \geq 2^{\theta + 1} y_{s_i}
  = 2^{\theta + 1} y_{t_i}
  > 2^{\theta+1}\frac 1{2^{\theta+1}} = 1.\]
  In both cases, all corresponding constraints in \ref{DST-LP} are satisfied.
\end{proof}

\begin{claim}
\label{claim:algo_density}
  The density of $T$ is at most $O(\alpha \log k) \sum_{e \in E} c(e) x_e^*$.
\end{claim}
\begin{proof}
  By \Cref{claim:algo_lp_feasible}, along with the fact that the
  DST algorithm used to construct $T_t$ and $T_s$ is an $\alpha$-approximation
  with respect to the optimal fractional solution, the costs of
  $T_t$ and $T_s$ are each upper bounded by
  $\alpha 2^{\theta + 1} \sum_{e \in E} c(e) x_e^*$.

  To bound the number of terminals covered by $T$, i.e. $|D_{\theta}|$,
  note that $\sum_{i \in D_\theta} y_{t_i} \leq \sum_{i \in D_{\theta}} 1/2^\theta
  = |D_\theta| 1/2^{\theta}$. Thus by \Cref{claim:algo_good_group},
  $|D_\theta| \geq 2^{\theta} \sum_{i \in D_{\theta}} y_{t_i}
  \geq 2^{\theta}/(2\log k + 2)$. Therefore, the density of $T$ is at most
  \begin{align*}
    &\frac{c(T_t) + c(T_s)}{|D_{\theta}|}
  \leq \frac {2\log k + 2} {2^{\theta}} (\alpha 2^{\theta + 2} \sum_{e \in E} c(e)x_e^*)
  = 8\alpha(\log k + 1) \sum_{e \in E} c(e) x_e^*.
  \end{align*}
\end{proof}

\Cref{lem:algo-main} follows from \Cref{claim:algo_density},
\Cref{claim:den-lp-valid}, and the
fact that $T$ is a valid junction tree.

\section{Proof of \Cref{thm:main}}
\label{sec:main_thm_proof}

\Cref{thm:existence-main} and \Cref{thm:algo-main} suffice to
conclude the proof of \Cref{thm:main} via a greedy
covering approach that is standard for covering problems such as Set Cover.

\begin{proof}[Proof of Theorem \ref{thm:main}]
  Let $(G, D)$ be an instance of planar-DSF. Combining
  \Cref{thm:existence-main} and \Cref{thm:algo-main}, we can find
  a junction tree of density at most $O(\log^5 k) \opt(G)/k$.
  Let $H_1$ be such a junction tree on $(G, D)$, and let $D_1$ be the set
  of terminal pairs covered by $H_1$. We remove $D_1$ from $D$ and repeat
  until all terminal pairs are covered.
  Since each junction tree covers at least one terminal pair, this process
  terminates in at most $k$ iterations. Let $H_1, \dots, H_{\ell}$ be the
  junction trees formed by this process, and for $j \in [\ell]$, let
  $D_j$ be the set of terminals covered by $H_j$. We denote by
  $D_{<j}$ the set $\cup_{j' < j} D_{j'}$.
  We then return $H = \cup_{j \in [\ell]} H_j$.

  It is clear by construction that $H$ is a feasible solution: for each
  $i \in [k]$, there exists some junction tree $H_j$ that covers $(s_i, t_i)$, and
  the path in $H_j$ from $s_i$ to its root concatenated with the path from
  the root to $t_i$ is an $s_i$-$t_i$ path in $H$. To bound the
  cost of $H$, note that
  $c(H) \leq \sum_{j \in [\ell]} c(H_j) = \sum_{j \in [\ell]} |D_j| \cdot \den(H_j)$.
  By construction, each $H_j$ has density at most $O(\log^5 k_j) \opt(G)/k_j$,
  where $k_j = |D \setminus D_{<j}|$ is the number of terminals remaining 
  when constructing $H_j$. Therefore,
  \begin{align*}
    c(H) \leq \sum_{j \in [\ell]} |D_j| \cdot O(\log^5 k_j) \frac{\opt(G)}{|D \setminus D_{<j}|}
    \leq O(\log^5 k) \opt(G) \sum_{j \in [\ell]}
    \frac{|D\setminus D_{<j}| - |D \setminus D_{< j+1}|}{|D \setminus D_{< j}|}.
  \end{align*}
  The term $\sum_{j \in [\ell]}
    \frac{|D\setminus D_{<j}| - |D \setminus D_{< j+1}|}{|D \setminus
      D_{< j}|}$
    is bounded by the $k$'th harmonic number $H_k$, thus $c(H)$ is at most
  $O(\log^6 k) \opt(G)$.
\end{proof}

\section{Conclusion}
\label{sec:conclusion}
Several open questions arise from our work in this paper.  It is
unlikely that the $O(\log^6 k)$ approximation ratio that we obtained
is tight. There are no known lower bounds that rule out a
constant-factor approximation for DSF in planar graphs. Closing this
gap is a compelling question.  Second, can we
establish a poly-logarithmic ratio upper bound on the integrality gap
of the natural cut-based LP relaxation for planar-DSF?  Our techniques
do not directly generalize to fractional solutions (see
\Cref{rem:dsf_integrality_gap}); however, we are hopeful that other
approaches may yield positive results.  Another direction for future
research is to extend this work and also the recent work on DST and related
problems from planar graphs to any proper minor-closed family of
graphs. Finally, there are several generalizations of DST and DSF that may
also admit positive results in planar graphs.

\paragraph{Acknowledgements:} We thank the anonymous reviewers for 
helpful suggestions and pointers to related work.

\bibliographystyle{plainurl}
\bibliography{planar_dsf,dst_extensions_paper}

\end{document}